\documentclass[11pt]{article}
\usepackage{amsmath,amssymb,amsfonts,amsthm,graphicx,fullpage}
\usepackage{url}

\urldef{\mailsaf}\path|fenner.sa@gmail.com|

\newcommand{\reals}{\mathbb{R}}
\newcommand{\rats}{\mathbb{Q}}
\newcommand{\nums}{\mathbb{N}}
\newcommand{\dyads}{{\rats_2}}
\newcommand{\two}{{\{0,1\}}}
\newcommand{\strs}{{\two^*}}
\newcommand{\seqs}{{\two^{\infty}}}
\newcommand{\map}[3]{{{#1}:{#2}\rightarrow{#3}}}
\newcommand{\prefeq}{\sqsubseteq}
\newcommand{\pref}{\sqsubset}
\newcommand{\emptystr}{\lambda}
\newcommand{\myand}{{\;\mathrel{\wedge}\;}}
\newcommand{\ceiling}[1]{{\lceil{#1}\rceil}}
\newcommand{\floor}[1]{{\lfloor{#1}\rfloor}}
\newcommand{\bigceiling}[1]{{\left\lceil{#1}\right\rceil}}

\newcommand{\assn}{\leftarrow}
\newcommand{\eps}{\varepsilon}
\newcommand{\bigabs}[1]{{\left|{#1}\right|}}
\newcommand{\EXP}{\textbf{EXP}}
\newcommand{\BPP}{\textbf{BPP}}
\newcommand{\tuple}[1]{{\langle{#1}\rangle}}
\newcommand{\union}{\mathop{\cup}}
\newcommand{\intersection}{\mathop{\cap}}

\newtheorem{definition}{Definition}[section]
\newtheorem{theorem}[definition]{Theorem}
\newtheorem{lemma}[definition]{Lemma}
\newtheorem{proposition}[definition]{Proposition}
\newtheorem{corollary}[definition]{Corollary}
\newtheorem{claim}[definition]{Claim}
\newtheorem{observation}[definition]{Observation}
\newtheorem{conjecture}[definition]{Conjecture}
\newtheorem{example}[definition]{Example}

\newenvironment{algo}{\begin{tabbing}
\hspace{0.25in}\=\hspace{0.25in}\=\hspace{0.25in}\=\hspace{0.25in}\=\hspace{0.25in}\=\hspace{0.25in}\=\hspace{0.25in}\=\hspace{0.25in}\=\kill}{\end{tabbing}}

\begin{document}


\title{Functions that preserve p-randomness\footnote{An extended abstract of this paper appeared in FCT~2011 \cite{Fenner:p-rand-func}.}}

\author{Stephen A. Fenner\thanks{Partially supported by NSF grants CCF-0515269 and CCF-0915948.}\\
%
Computer Science and Engineering Department\\
University of South Carolina\\
Columbia, SC 29208 USA\\
\mailsaf}

\maketitle

\begin{abstract}
We show that polynomial-time randomness (p-randomness) is preserved under a variety of familiar operations, including addition and multiplication by a nonzero polynomial-time computable real number.  These results follow from a general theorem: If $I\subseteq\reals$ is an open interval, $\map{f}{I}{\reals}$ is a function, and $r\in I$ is p-random, then $f(r)$ is p-random provided
\begin{enumerate}
\item
$f$ is p-computable on the dyadic rational points in $I$, and
\item
$f$ varies sufficiently at $r$, i.e., there exists a real constant $C > 0$ such that either
\[ (\forall x\in I-\{r\})\left[\frac{f(x) - f(r)}{x-r} \ge C\right] \]
or
\[ (\forall x\in I-\{r\})\left[\frac{f(x) - f(r)}{x-r} \le -C\right]\;. \]
\end{enumerate}

Our theorem implies in particular that any analytic function about a p-computable point whose power series has uniformly p-computable coefficients preserves p-randomness in its open interval of absolute convergence.  Such functions include all the familiar functions from first-year calculus.

\bigskip

\end{abstract}

\noindent\textbf{Subject Classification:} Computational complexity



\section{Introduction}

Informally, we might call an infinite binary sequence ``random'' if we see no predictable patterns in the sequence.  Put another way, a sequence is random if it looks ``typical,'' that is, it enjoys no easily identifiable properties not shared by almost all other sequences.  Here, the notion of ``almost all'' comes from Lebesgue measure on the unit interval $[0,1]$.  What we mean by ``easily identifiable,'' on the other hand, can vary greatly with the situation.  In statistics, random sequences are useful to avoid bias in sampling or in simulating processes (e.g., queueing systems) that are too complex for us to determine exactly.  In statistics, desirable properties for random sequences include instances of the law of large numbers: a fixed sequence of length $n$ should occur in the sequence asymptotically a $2^{-n}$ fraction of the time, for example.  Other examples include the law of the iterated logarithm.  In cryptography and network security, ``easily identifiable'' must be strengthened to ``unpredictable by an adversary.''  In computer science generally, random sequences should produce successful results most of the time when used in various randomized algorithms.

There is always a trade-off between the amount of randomness possessed by a sequence and the ease with which it can be produced.  Random sequences that can be produced algorithmically (i.e., pseudorandom sequences) are of course desirable, provided they have enough randomness for the task at hand.  The study of algorithmic randomness has a long and rich history (see, for example, \cite{DH:randomness,DHNT:randomness} for references to the literature).  Complexity theoretic notions of randomness were first suggested by Schnorr, and resource-bounded measure and randomness were developed more fully by Lutz (see \cite{Lutz:measure-survey}).  For a survey on the subject, see \cite{AM:randomness-survey}.

A natural trade-off in the context of polynomial-time computation is the notion of polynomial-time randomness, or p-randomness for short (see Definition~\ref{def:p-random}, below), which is closely tied with the notion of p-measure introduced by Lutz \cite{Lutz:thesis,Lutz:measure}.  There are p-random sequences that can be computed in exponential time; in fact, almost all sequences in $\EXP$ (in a resource-bounded measure theoretic sense) are p-random.  Yet p-random sequences are still strong enough for many common tasks, both statistical and computational.  For example, p-random sequences satisfy the laws of large numbers and the iterated logarithm (see \cite{Wang:randomness}), and they provide adequate sources for $\BPP$ computations and have many other desirable computational properties (see \cite{Lutz:measure-survey}).

The current work addresses some geometric aspects of p-random sequences.  Recently, connections between the geometry of Euclidean space and effective and resource-bounded measure and dimension have been found \cite{LM:fractals,LW:connectivity}.  The question of how the complexity or measure theoretic properties of a real number are altered when it is transformed via a real-valued function goes back at least to Wall~\cite{Wall:normal}, who showed that adding or multiplying a nonzero rational number to a real number whose base-$k$ expansion is normal\footnote{An infinite sequence $s$ over a $k$-letter alphabet $\Sigma$ is \emph{normal} iff for any finite string $w\in\Sigma^*$, there are $nk^{-|w|}(1+o(1))$ occurrences of $w$ as a substring among the first $n$ letters of $s$, as $n$ tends to infinity.}  yields another real with a normal base-$k$ expansion.  Doty, Lutz, \& Nandakumar recently extended Wall's result, showing that the finite-state dimension of the base-$k$ expansion of a real number is preserved under addition or multiplication by a nonzero rational number \cite{DLN:arithmetic}.  At the other extreme of the complexity spectrum, it is not hard to show that algorithmic randomness (Martin-L\"of randomness \cite{MartinLoef:random}) is preserved under addition or multiplication by a nonzero computable real, regardless of the base of the expansion.

In this paper we take a middle ground, considering how polynomial-time computable functions mapping reals to reals preserve p-randomness.  We show (Theorem~\ref{thm:main}, below) that such a function $f$ maps a p-random real $r$ to a p-random real $f(r)$ provided $f$ satisfies a kind of anti-Lipschitz condition in some neighborhood of $r$: $f(x)$ varies from $f(r)$ at least linearly in $x-r$.  (This result still holds even if $f$ is not monotone in any neighborhood of $r$, or if $f$ is only polynomial-time computable on dyadic rational inputs, or if $f$ enjoys no particular continuity properties.)

Our result has a number of corollaries: p-randomness is preserved under addition and multiplication by nonzero p-computable reals (complementing the results in \cite{Wall:normal,DLN:arithmetic} and the folklore result about algorithmically random reals); it is also preserved by polynomial and rational functions (with p-computable coefficients) and all the familiar transcendental functions on the reals, e.g., exponential, logarithmic, and trigonometric functions.

The polynomial-time case presents some technical challenges not present with unbounded computational resources.  Roughly speaking, given a polynomial-time approximable function $\map{f}{\reals}{\reals}$, our goal is to define a betting strategy (i.e., a martingale; see Section~\ref{sec:basic}) that bets on the next bit of the binary expansion of a real number $r$, given previous bits.  The strategy is based on the behavior of an assumed strategy $d$ that successfully bets on $f(r)$.  If we had no resource bounds, then we could approximate $f$ at various points as closely as needed to obtain a good sample of $d$'s behavior on $f$ applied to those points, allowing us to mimic $d$ and thus succeed on $r$.  Since we are polynomial-time-bounded, however, we have no such luxury, and we have to settle for rougher approximations of $f$.  For example, $d$ could succeed on $f(0.0111111111\cdots)$ (where there is a long string of $1$'s before the next $0$ in the argument to $f$) but lose everything on $f(0.10000\cdots)$, which is close by.  If we only have a poor approximation to $f$, then we cannot distinguish the two cases above, and so $d$ is no good at telling us how to bet on the first digit after the decimal point.  Fortunately, we may assume that $d$ is conservative---in the sense that it does not bet drastically---so that $d$'s assets are relatively insensitive to slight variations in the real numbers corresponding to the sequences it bets on.  


\bigskip

Section~\ref{sec:basic} has basic definitions, including martingales and p-randomness.  Section~\ref{sec:functions} describes the conditions on real-valued functions sufficient to preserve p-randomness.  Our main results are in Section~\ref{sec:main}, where we prove that these conditions indeed suffice; Theorem~\ref{thm:main} is the main result of that section.  In Section~\ref{sec:apps}, we show that these conditions hold for a variety of familiar functions.  In Section~\ref{sec:tight}, we give evidence that the strongly varying hypothesis in Theorem~\ref{thm:main} is tight.  In Section~\ref{sec:measure}, we provide a result about p-measure that is analogous to our main result about p-randomness.  We suggest further research in Section~\ref{sec:open}.


\section{Notation and basic definitions}
\label{sec:basic}

We let $\nums = \{0,1,2,\ldots\}$.  We let $\rats$ be the set of rational numbers.  A \emph{dyadic rational} is some $q\in\rats$ expressible as $\pm a/2^b$ for some $a,b\in\nums$.  We let $\dyads$ denote the set of dyadic rational numbers.

For real $x>0$, we let $\lg x$ denote $\log_2 x$.

In this paper, we only consider the binary expansions of real numbers.  If need be, all our results can easily be modified to other bases.

Our basic notions and results about p-computability, martingales, and randomness in complexity theory are standard.  See, for example, \cite{Lutz:measure,Lutz:measure-survey,AM:randomness-survey}.

Let $w\in\strs$ and $s\in\seqs$.  We let $|w|$ denote the length of $w$, and for any $0\le i < |w|$ we let $w[i]$ be the $(i+1)$st bit of $w$.  Similarly, for any $i\in\nums$ we let $s[i]$ denote the $(i+1)$st bit of $s$.  For any $m,n\in\nums$ with $m\le n$, we let $s[m\ldots(n-1)] = s[m]s[m+1]\cdots s[n-1] \in\strs$ denote the substring consisting of the $(m+1)$st through the $n$th bit of $s$.  We let $\two^n$ denote the set of strings in $\strs$ of length $n$.  If $v\in\strs\union\seqs$, we let $w\prefeq v$ mean that $w$ is a prefix of $v$, and we let $w\pref v$ mean that $w$ is a proper prefix of $v$.  We denote the empty string by $\emptystr$.

Recall that a \emph{martingale} is a function $\map{d}{\strs}{\reals}$ such that for every $w\in\strs$,
\[ 0 \le d(w) = \frac{d(w0) + d(w1)}{2}\;. \]
We will also assume without loss of generality that $d(\emptystr) \le 1$.  We say that $d$ \emph{succeeds} on a sequence $s\in\seqs$ iff
\[ \limsup_{n\rightarrow\infty} d(s[0\ldots(n-1)]) = \infty\;. \]
We say that $d$ \emph{strongly succeeds} on $s$ iff
\[ \liminf_{n\rightarrow\infty} d(s[0\ldots(n-1)]) = \infty\;. \]

\begin{definition}\rm
Fix any $k\in\nums$.  A function $\map{d}{\strs}{\reals}$ is \emph{$n^k$-computable} if there is a function $\map{\hat{d}}{\strs\times\strs}{\rats}$ such that
\[ \left| d(w) - \hat{d}(w,0^r) \right| \le 2^{-r} \]
for every $w\in\strs$ and $r\in\nums$, and in addition, $\hat d(w,0^r)$ is computable in time $O((|w|+r)^k)$.  We say that $\hat{d}$ is a \emph{$n^k$-approximator} for $d$.  We say that $d$ is \emph{p-computable} if $d$ is $n^k$-computable for some $k$, and that $\hat d$ is a \emph{p-approximator} for $d$ if $\hat d$ is an $n^k$-approximator for $d$, for some $k$.  A real number $c$ is \emph{$n^k$-computable} (respectively, \emph{p-computable}) if the constant function $\strs\rightarrow\{c\}$ is $n^k$-computable (respectively, p-computable), and we may suppress the first argument in a p-approximator for $c$.
\end{definition}

\begin{definition}\label{def:p-random}\rm
Let $s\in\seqs$ be any sequence.
\begin{enumerate}
\item
For any $k\in\nums$, $s$ is \emph{$n^k$-random} if no $n^k$-computable martingale succeeds on $s$.
\item
The sequence $s\in\seqs$ is \emph{p-random} if $s$ is $n^k$-random for all $k$, i.e., no p-computable martingale succeeds on $s$.
\end{enumerate}
\end{definition}

\begin{definition}\rm
We will say that a martingale $d$ is \emph{conservative} iff
\begin{enumerate}
\item
for any $w\in\strs$ and $b\in\two$,
\[ \frac{d(w)}{2} \le d(wb) \le \frac{3d(w)}{2}\;, \]
and
\item
for any $s\in\seqs$, if $d$ succeeds on $s$, then $d$ strongly succeeds on $s$.
\end{enumerate}
\end{definition}

Note that if $d$ is conservative, then $d(w) \le (3/2)^{|w|}$ for all $w$.  It is well-known (and easy to show) that if there is a
p-computable martingale that succeeds on $s$, then there is a conservative p-computable martingale that succeeds on $s$.  Moreover, there is a bound on the running time of the conservative martingale that depends only on the running time of the original martingale (and not on the martingale itself or on $s$).  More precisely,

\begin{proposition}\label{prop:conservative}
For any $k\in\nums$, there exists $\ell\in\nums$ such that, for any $n^k$-computable martingale $d$, there exists a conservative $n^\ell$-computable martingale $d'$ that (strongly) succeeds on every sequence that $d$ succeeds on.
\end{proposition}

We identify a sequence $s\in\seqs$ with a real number $0.s\in[0,1]$ via the usual binary expansion: $0.s := \sum_{i=0}^\infty s[i]2^{-(i+1)}$.  This correspondence is one-to-one except on $\dyads$, where it is two-to-one.  For every $x\in\strs$, we define $0.x := 0.x000\cdots$, and we define the \emph{dyadic interval}
\[ \Gamma_x := [0.x,0.x+2^{-|x|}] = \{ 0.s: s\in\seqs \myand x\pref s\}\;. \]


For $s\in\seqs$, we define $0.s$ to be p-random (respectively, $n^k$-random) iff $s$ is p-random (respectively, $n^k$-random).  If $x\in\reals$, then we define $x$ to be p-computable (p-random) just as we do for $x - \floor x$, and similarly for $O(n^k)$ computability and $n^k$-randomness.  It is well-known that no p-computable real number is p-random.

\section{Functions of interest}
\label{sec:functions}


We will restrict our attention to certain types of real-valued functions of a real variable.  We are only interested in the behavior of these functions on p-random inputs.  For simplicity, we will only consider functions with domain $[0,1]$, but this is in no way an essential restriction.  Our functions will possess a certain p-computability property and a certain strong variation property.  Both these properties are \emph{local} in the sense that we only care about them in the vicinity of a p-random number.

\begin{definition}\rm
A function $\map{f}{[0,1]}{\reals}$ is \emph{weakly p-computable} if there exists a polynomial-time computable function $\map{\hat{f}}{\strs\times\strs}{\rats}$ such that for any $w\in\strs$ and $r\in\nums$,
\[ \left| \hat{f}(w,0^r) - f(0.w) \right| \le 2^{-r}\;. \]
Furthermore, for constant $k\in\nums$, if $\hat f(w,0^r)$ is computable in time $O((|w|+r)^k)$, then we say that $f$ is \emph{weakly $n^k$-computable}.
\end{definition}

Note that a weakly p-computable function can behave arbitrarily on $[0,1] - \dyads$.

\begin{definition}\label{def:weak-p-comp-at-x}\rm
Let $\map{f}{[0,1]}{\reals}$ be a function and let $\Gamma_y\subseteq [0,1]$ be some dyadic interval with $y\in\strs$.  We say that $f$ is \emph{weakly p-computable on $\Gamma_y$} iff there exists a ptime computable function $\map{\hat{f}}{\strs\times\strs}{\rats}$ such that for any $w\in\strs$ and $r\in\nums$,
\[ \left| \hat{f}(w,0^r) - f(0.(yw)) \right| \le 2^{-r}\;. \]
If $x\in [0,1]$, then we say that \emph{$f$ is weakly p-computable at $x$} iff $f$ is weakly p-computable on some dyadic interval containing $x$.

All these notions carry over in the obvious way when ``p-computable'' is replaced with ``$n^k$-computable.''

We say that $f$ is \emph{locally weakly p-computable} if $f$ is weakly p-computable at all p-random points in $[0,1]$.
\end{definition}

[Note that $0.(yw) \in \dyads$ is the dyadic rational number corresponding to the string $yw$ (the concatenation of $y$ and $w$).]


In other words, $f$ is weakly p-computable at $x$ iff we can approximate $f$ on the dyadic rationals in some dyadic interval containing $x$ in polynomial time.  Notice that we are \emph{not} insisting that $f$ have any continuity properties.  This means in particular that $\hat{f}$ may not uniquely determine $f$ on $\Gamma_x$.  Notice also that a function may be locally weakly p-computable but not ``globally'' p-computable, being patched together nonuniformly with various p-computable functions on different dyadic intervals.

We can extend Definition~\ref{def:weak-p-comp-at-x} to weak p-computability at an arbitrary point $x\in\reals$ in the natural way.

\begin{definition}\label{def:strongly-varies}\rm
Let $I\subseteq \reals$ be an interval, let $\map{f}{I}{\reals}$ be a function, and let $x\in I$ be some point.  We say that $f$ \emph{strongly varies at $x$ on $I$} iff there is some real constant $C>0$ such that either
\begin{enumerate}
\item
for all $z\in I - \{x\}$,
\[ \frac{f(z) - f(x)}{z-x} \ge C\;, \]
or
\item
for all $z\in I$,
\[ \frac{f(z) - f(x)}{z-x} \le -C\;. \]
\end{enumerate}
In case~(1) we say that \emph{$f$ strongly increases at $x$ on $I$}, and in case~(2) \emph{$f$ strongly decreases at $x$ on $I$}.

We say that \emph{$f$ strongly varies at $x$} if $f$ strongly varies at $x$ on $N$ for some open interval $N$ containing $x$.  We define $f$ strongly increasing/decreasing at $x$ analogously.
\end{definition}




The notion of strong variation is illustrated in Figure~\ref{fig:strong-variation}.
\begin{figure}
\centering
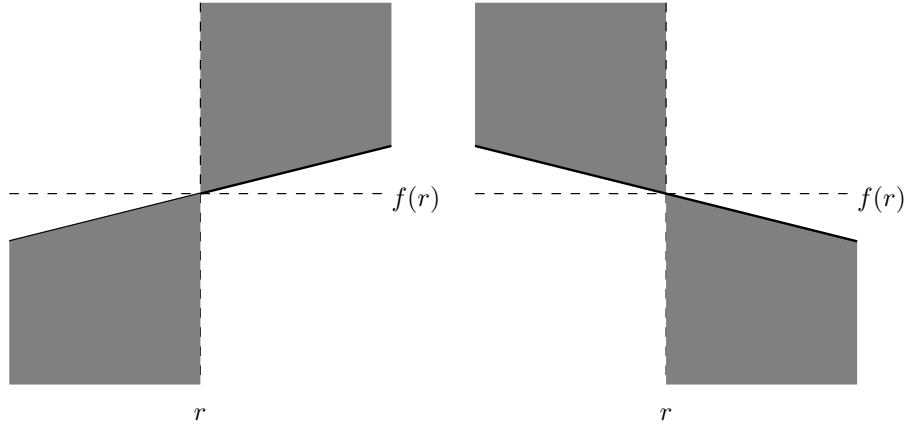
\caption{For $f$ to strongly vary at $r$, its graph must confine itself to the shaded region on the left (if strongly increasing) or the right (if strongly decreasing) in some neighborhood of $r$.  The thick line on the left has slope $C$ (satisfying $y-f(r) = C(x-r)$), and the line on the right has slope $-C$ (satisfying $y-f(r) = -C(x-r)$), for some constant $C>0$.  Both diagrams depict an arbitrarily small neighborhood of $r$.}\label{fig:strong-variation}
\end{figure}

\begin{example}\label{ex:nonzero-derivative}
If $f$ is $C^1$ in a neighborhood of $x$ and $f'(x) \ne 0$, then $f$ strongly varies at $x$.
\end{example}

\section{Main result}
\label{sec:main}

Here is our main technical theorem, from which most of the other results in the paper follow easily.

\begin{theorem}\label{thm:main}
Let $I\subseteq \reals$ be some interval and $\map{f}{I}{\reals}$ some function.  Suppose $r$ is a p-random point in the interior of $I$.  If $f$ is weakly p-computable at $r$ and strongly varies at $r$, then $f(r)$ is p-random.
\end{theorem}




\subsection{Establishing Theorem~\ref{thm:main}}

We start this section with two easy observations which we give without proof.

\begin{observation}\label{obs:scale-shift-random}
Let $j$ and $k$ be integers with $k\ge 0$, and let $a\in\dyads$.  A number $r\in\reals$ is $n^k$-random if and only if $2^j r$ is $n^k$-random, if and only if $r+a$ is $n^k$-random, if and only if $-r$ is $n^k$-random.

The same then obviously holds when ``$n^k$-random'' is replaced with ``p-random.''
\end{observation}

\begin{observation}\label{obs:scale-shift}
Let $I\subseteq\reals$ be an interval, let $\map{f}{I}{\reals}$ be a function, let $j$ and $k$ be any integers with $k\ge 0$, and let $a\in\dyads$.  Define
\begin{align*}
g(x) &=  2^j f(x)\;, \\
h(x) &=  f(x) + a\;, \\
j(x) &=  f(2^j x)\;, \\
k(x) &=  f(x + a)\;.
\end{align*}
Then $f$ strongly varies at some $x\in I$ on $I$ (respectively, is $n^k$-computable at $x$) if and only if all of $(-f),g,h$ strongly vary (respectively, are $n^k$-computable) at $x$ on $I$, if and only if $j$ strongly varies (respectively, is $n^k$-computable) at $2^{-n} x$ on $2^{-n}I$, if and only if $k$ strongly varies (respectively, is $n^k$-computable) at $x-a$ on $I-a$.  The sense of variation (strongly increasing or strongly decreasing) of $f$ is the same as that of $g,h,j,k$ and opposite that of $(-f)$.

The same then obviously holds when ``$n^k$-computable'' is replaced with ``p-computable.''
\end{observation}

Theorem~\ref{thm:main} is a corollary of the next lemma, which gives the theorem its essential technical content.  We prove this lemma later in this section.  For convenience, we will assume that our function $f$ is monotone ascending.  We will show later that this is not an essential restriction.

\begin{lemma}\label{lem:main}
For any $j,k\in\nums$ there exists $q\in\nums$ such that, for any weakly $n^j$-computable, monotone ascending $\map{f}{[0,1]}{\reals}$ and $x_0\in [0,1]$ such that $f$ strongly increases at $x_0$ on $[0,1]$, if $f(x_0)$ is not $n^k$ random, then $x_0$ is not $n^q$-random.
\end{lemma}

The full strength of Lemma~\ref{lem:main} will only be used in Section~\ref{sec:measure}.  For the rest of the paper, we can content ourselves with the following corollary:

\begin{corollary}\label{cor:main}
Let $\map{f}{[0,1]}{\reals}$ be weakly p-computable and monotone ascending on $[0,1]$.  Suppose that $x_0\in [0,1]$ and that $f$ strongly increases at $x_0$ on $[0,1]$.  Then if $f(x_0)$ is not p-random, then $x_0$ is not p-random.
\end{corollary}

To prove Lemma~\ref{lem:main}, we need to construct an $n^q$-computable martingale $d_f$ that succeeds on $x_0$, given an $n^k$-computable one that succeeds on $f(x_0)$.  If martingale $d$ succeeds on $f(x_0)$, then we can define $d_f(w)$ (for a given string $w$) to sample the values of $d$ on points in $f(\Gamma_w)$.  We do this by sandwiching $d_f(w)$ between a lower bound $d^-(w;n)$ and an upper bound $d^+(w;n)$.  We get $d^+(w;n)$ by overestimating $d$'s total contribution in an interval around $f(0.w)$ (Equation~(\ref{eqn:upper-f-shift}), below), and we get $d^-(w;n)$ by underestimating it (Equation~(\ref{eqn:lower-f-shift})).  These estimates become more refined as $n$ increases, and, provided $d$ is conservative, they reach a common limit as $n$ goes to infinity, yielding a well-defined martingale $d_f$.

\begin{definition}\label{def:f-shifts}\rm
Let $\map{f}{[0,1]}{[0,1]}$ be monotone ascending on $[0,1]$ and let $d$ be a martingale.  For every $x\in\strs$, let $\Delta_x$ denote the interval $f(\Gamma_x) = [f(0.x),f(0.x+2^{-|x|})]$, and for every $n\in\nums$, define
\begin{equation}\label{eqn:upper-f-shift}
d^+(x;n) = 2^{|x|-n} \sum_{y\in\two^n\;:\;\Gamma_y \intersection \Delta_x\ne\emptyset} d(y)\;,
\end{equation}
and define
\begin{equation}\label{eqn:lower-f-shift}
d^-(x;n) := 2^{|x|-n} \sum_{y\in\two^n\;:\;\Gamma_y \subseteq \Delta_x} d(y)\;.
\end{equation}
\end{definition}

The only differences between the sums in Equations~(\ref{eqn:upper-f-shift}) and (\ref{eqn:lower-f-shift}) are at most two terms $d(y)$ where $\Gamma_y$ straddles the boundary of $\Delta_x$.  The assumption that $d$ is conservative is needed to ensure that these terms are not too large, and thus that $d^+(x;n)$ and $d^-(x;n)$ are close to each other.  The following lemma is routine and easy to check.

\begin{lemma}\label{lem:monotone}
Let $f$ and $d$ be as in Definition~\ref{def:f-shifts}.  For any $x\in\strs$, if $\Gamma_y$ is any dyadic interval contained in $\Delta_x$ (that is, $\Gamma_y \subseteq \Delta_x$), then letting $n = |y|$,
\begin{align*}
2^{|x|-n} d(y) \le d^-(x;n) &\le d^-(x;n+1) \le d^-(x;n+2) \le \cdots \le d^-(x;n+i) \le \cdots \\
 \cdots &\le d^+(x;i) \le \cdots \le d^+(x;2) \le d^+(x;1) \le d^+(x;0)\;.
\end{align*}
\end{lemma}

\begin{proof}
The first inequality holds because $\Gamma_y \subseteq \Delta_x$, and hence $2^{|x|-n} d(y)$ is one of the terms in the sum defining $d^-(x;n)$.  To see the other inequalities on the top line, notice that each term $2^{|x|-(n+i)}d(y)$ in the expression for $d^-(x;n+i)$ (for some $i\in\nums$) is equal to the sum $2^{|x|-(n+i+1)} d(y0) + 2^{|x|-(n+i+1)} d(y1)$ of two terms occurring in the expression for $d^-(x;n+i+1)$.  This follows from the fact that any $\Gamma_y\subseteq \Delta_x$ contains both $\Gamma_{y0}$ and $\Gamma_{y1}$, and so the latter two intervals are also subsets of $\Delta_x$.

Clearly, all terms in the sum for $d^-(x;n+i)$ are included in the sum for $d^+(x;n+i)$, and so every quantity on the top line is less than or equal to the corresponding quantity on the bottom line.

Finally, the inequalities on the bottom line all hold: If we split each term $2^{|x|-i}d(y)$ in the expression for $d^+(x;i)$ into the equivalent sum
\[ 2^{|x|-(i+1)} d(y0) + 2^{|x|-(i+1)} d(y1)\;, \]
then this accounts for all the terms in $d^+(x;i+1)$ (and possibly more).
\end{proof}

\begin{definition}\rm
Let $f$ and $d$ be as in Definition~\ref{def:f-shifts}.  We define the \emph{upper $f$-shift} of $d$ to be the function defined for all $x\in\strs$ as
\[ d^+(x) := \lim_{n\rightarrow\infty} d^+(x;n)\;. \]
Similarly, we define the \emph{lower $f$-shift} of $d$ to be
\[ d^-(x) := \lim_{n\rightarrow\infty} d^-(x;n)\;. \]
\end{definition}

Since for any fixed $x\in\strs$, $d^+(x;n)$ and $d^-(x;n)$ are both monotone functions of $n$ (decreasing and increasing, respectively) by Lemma~\ref{lem:monotone}, the limits in the definition above clearly exist, and
\[ d^-(x;n) \le d^-(x) \le d^+(x) \le d^+(x;n) \]
for all $n$.

For some martingales, the upper and lower $f$-shifts may differ, but they coincide for conservative martingales.

%

\begin{lemma}\label{lem:squeeze}
Fix $f$ and $d$ as in Definition~\ref{def:f-shifts}.  Suppose further that $d$ is conservative.  For any $x\in\strs$ and $n\in\nums$,
\begin{equation}\label{eqn:squeeze}
d^+(x;n) - d^-(x;n) \le
2^{|x|+1}\left(\frac{3}{4}\right)^n\;.
\end{equation}
\end{lemma}

\begin{proof}
Here we only use Property~(1) of being conservative.  All the terms in the two sums on the left-hand side of the inequality~(\ref{eqn:squeeze}) cancel except for at most two dyadic intervals $\Gamma_{y_\textup{left}}$ and $\Gamma_{y_\textup{right}}$---the former containing the left endpoint of $\Delta_x$ and the latter containing the right endpoint.  Thus we get
\begin{align*}
d^+(x;n) - d^-(x;n) &\le 2^{|x|-n}(d(y_\textup{left}) + d(y_\textup{right})) \le 2^{|x|-n}\left[\left(\frac{3}{2}\right)^n + \left(\frac{3}{2}\right)^n\right] \\
&= 2^{|x|+1}\left(\frac{3}{4}\right)^n\;.
\end{align*}
\end{proof}

\begin{corollary}\label{cor:c-shift}
Let $f$ and $d$ be as in Definition~\ref{def:f-shifts}.  If $d$ is conservative, then $d^+(x) = d^-(x)$ for all $x\in\strs$.
\end{corollary}

\begin{proof}
Immediate from Lemma~\ref{lem:squeeze}.
\end{proof}

\begin{definition}\rm
If $f$ and $d$ are as in Definition~\ref{def:f-shifts} and $d$ is conservative, then we let $d_f(x)$ denote the common value $d^+(x) = d^-(x)$, and we call $d_f$ the \emph{$f$-pullback} of $d$.
\end{definition}

On input string $x$, $d_f(x)$ merely samples $d$ over the the interval $\Delta_x = f(\Gamma_x)$.

%

\begin{lemma}\label{lem:is-a-martingale}
If $f$ and $d$ are as in Definition~\ref{def:f-shifts} and $d$ is conservative, then its $f$-pullback $d_f$ is a martingale.
\end{lemma}

\begin{proof}
To see that $d_f$ is a martingale, first we notice that
\[ 0 \le d_f(\emptystr) \le d^+(\emptystr;0) = d(\emptystr) \le 1\;. \]
Next, by examining terms in the sums and using Lemma~\ref{lem:monotone}, we notice that for any $x\in\strs$ and $n\in\nums$,
\[ d^-(x0;n) + d^-(x1;n) \le 2d^-(x;n) \le 2d^+(x;n) \le d^+(x0;n) + d^+(x1;n)\;. \]
Taking the limit of all sides as $n\rightarrow\infty$, we get
\[ d^-(x0) + d^-(x1) \le 2d^-(x) \le 2d^+(x) \le d^+(x0) + d^+(x1)\;. \]
All these quantities are equal, since the two extremes are equal.  Thus
\[ d_f(x) = \frac{d_f(x0) + d_f(x1)}{2}\;. \]
\end{proof}

The next lemma is key.  Here is where we make essential use of the strongly increasing property of $f$.  (The hypothesis here is slightly weaker, though).

\begin{lemma}\label{lem:key}
Let $f$ and $d$ be as in Definition~\ref{def:f-shifts} with $d$ being conservative.  Suppose that there exist $r,s\in\seqs$ and a real $C>0$ such that
\begin{equation}\label{eqn:strong-increase}
\frac{f(x) - 0.r}{x-0.s} \ge C
\end{equation}
for all $x\in [0,1]-\{0.s\}$.  If $d$ succeeds on $r$ and $0.s\notin\dyads$, then $d_f$ succeeds on $s$.
\end{lemma}

\begin{proof}
Note that Equation~(\ref{eqn:strong-increase}) implies $f(x) < 0.r$ if $x<0.s$ and $f(x) > 0.r$ if $x>0.s$.

Set $\ell := \max(0,\ceiling{\lg(1/C)})$.  We then have $C \ge 2^{-\ell}$.

Since $0.s\notin\dyads$, $s$ has infinitely many $0$'s and infinitely many $1$'s.  This implies that $s$ has infinitely many occurrences of ``$01$'' as a substring, that is, there are infinitely many $n\in\nums$ such that $s[n]s[n+1] = 01$.  Fix any real $M > 0$.  Since $d$ succeeds on $r$ and is conservative, $d$ strongly succeeds on $r$, and so there is some $n_0\in\nums$ such that $d(r[0\ldots(n-1)]) \ge M$ for all $n\ge n_0$.  Fix some $n\ge n_0$ such that $s[n]s[n+1] = 01$.  Let $x = s[0\ldots(n-1)]$.  We have $|x|=n$ and $x01\pref s$.  Let $y = r[0\ldots(n+\ell+1)]$ be the first $n+\ell+2$ bits of $r$, noting that $d(y) \ge M$.  Here is the situation at $0.r$:
\begin{center}
\begin{picture}(0,0)%
\includegraphics{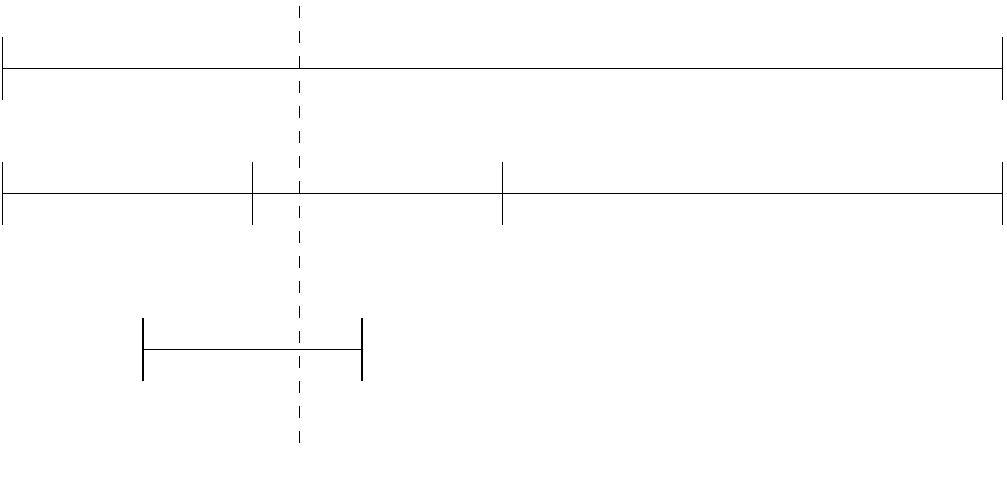}%
\end{picture}%
\setlength{\unitlength}{3947sp}%
\begingroup\makeatletter\ifx\SetFigFont\undefined%
\gdef\SetFigFont#1#2#3#4#5{%
  \reset@font\fontsize{#1}{#2pt}%
  \fontfamily{#3}\fontseries{#4}\fontshape{#5}%
  \selectfont}%
\fi\endgroup%
\begin{picture}(4824,2317)(1189,-4166)
\put(4801,-2686){\makebox(0,0)[b]{\smash{{\SetFigFont{10}{12.0}{\rmdefault}{\mddefault}{\updefault}$\Delta_{x1}$}}}}
\put(3001,-2686){\makebox(0,0)[b]{\smash{{\SetFigFont{10}{12.0}{\rmdefault}{\mddefault}{\updefault}$\Delta_{x01}$}}}}
\put(3601,-2086){\makebox(0,0)[b]{\smash{{\SetFigFont{10}{12.0}{\rmdefault}{\mddefault}{\updefault}$\Delta_x$}}}}
\put(2626,-4111){\makebox(0,0)[b]{\smash{{\SetFigFont{10}{12.0}{\rmdefault}{\mddefault}{\updefault}$0.r$}}}}
\put(2401,-3436){\makebox(0,0)[b]{\smash{{\SetFigFont{10}{12.0}{\rmdefault}{\mddefault}{\updefault}$\Gamma_y$}}}}
\put(1801,-2686){\makebox(0,0)[b]{\smash{{\SetFigFont{10}{12.0}{\rmdefault}{\mddefault}{\updefault}$\Delta_{x00}$}}}}
\end{picture}%

\end{center}
We have $0.r\in\Gamma_y$.  Also, since $0.s \in \Gamma_{x01} - \dyads$, we have $0.x01 < 0.s < 0.x01 + 2^{-|x01|}$, which implies $f(0.x01) \le 0.r \le f(0.x01+2^{-|x01|})$, as noted above.  It follows immediately that $0.r \in \Delta_{x01}$.

\begin{claim}\label{clm:subset}
$\Gamma_y \subseteq \Delta_x$.
\end{claim}

\begin{proof}[Proof of Claim~\ref{clm:subset}]
By Equation~(\ref{eqn:strong-increase}) we have
\[ 0.r - f(0.x) \ge C(0.s - 0.x) \ge C(0.x01 - 0.x) = C 2^{-(n+2)} \ge 2^{-(n+\ell+2)}\;. \]
Since $0.r \in \Gamma_y$, we have
\[ 0.r - 0.y \le 2^{-|y|} = 2^{-(n+\ell+2)}\;. \]
Combining these two inequalities gives $0.r - 0.y \le 0.r - f(0.x)$, or equivalently, $f(0.x) \le 0.y$.  Similarly, we have
\begin{align*}
f(0.x+2^{-n}) - 0.r & \ge C((0.x+2^{-n})-0.s) \ge C(0.x11 - 0.x1)  \\
&= C 2^{-(n+2)} \ge 2^{-(n+\ell+2)} =  2^{-|y|} \ge  0.y + 2^{-|y|} - 0.r\;,
\end{align*}
whence $0.y + 2^{-|y|} \le f(0.x + 2^{-n})$.  Thus
\[ \Gamma_y = [0.y,0.y+2^{-|y|}] \subseteq [f(0.x),f(0.x+2^{-n})] = \Delta_x \]
as claimed.  This concludes the proof of Claim~\ref{clm:subset}.
\end{proof}

Continuing with the proof of Lemma~\ref{lem:key}, we use Lemma~\ref{lem:monotone} again, noting that $|y| = |x| + \ell+2$, to get
\[ 2^{-(\ell+2)} d(y) = 2^{|x| - (|x|+\ell+2)} d(y) \le d^-(x;|x|+\ell+2) \le d_f(x)\;. \]
Since $d(y) \ge M$, we then get
\[ d_f(x) \ge 2^{-(\ell+2)}M\;. \]
Since $M$ is arbitrary, $x\sqsubseteq s$, and $\ell$ is a constant independent of $x$, this means that $d_f$ succeeds on $s$.
\end{proof}

Finally, we need a lemma regarding the computation time of $d_f$.  The challenge in the proof is in finding an easy (i.e., polynomial-time) way to approximate the $d^-(x;n)$ and $d^+(x;n)$ values.

\begin{lemma}\label{lem:p-computable}
For any $j,k\in\nums$ there exists $q\in\nums$ such that, for any conservative $n^k$-computable martingale $d$ and $n^j$-computable function $\map{f}{[0,1]}{[0,1]}$ monotone ascending on $[0,1]$ with $f(1)=1$, the $f$-pullback martingale $d_f$ of $d$ is $n^q$-computable.  (As a consequence, if $d$ is p-computable and $f$ is weakly p-computable on $[0,1]$, then $d_f$ is p-computable.)
\end{lemma}

\begin{proof}
The idea is that, given input $x\in\strs$ of length $n$ and accuracy parameter $r\in\nums$, we will approximate some number between $d^-(x;m)$ and $d^+(x;m)$ for some sufficiently large $m\ge n$ (but still polynomial in $n$ and $r$).  We have no hope of computing the sum of Equations~(\ref{eqn:upper-f-shift}) or (\ref{eqn:lower-f-shift}) directly, as there are exponentially many terms.  Fortunately, large blocks of the sum can be computed all at once by evaluating $d$ on shorter inputs.  The condition that $f(1)=1$ is only for technical convenience and is not necessary; it is only required that $f(1)$ be computable in time $O(n^j)$.

Given conservative $d$ and monotone $f$ as above, fix approximators
\[
\map{\hat{d}}{\strs\times\strs}{\rats} \hspace{0.25in}\mbox{and}\hspace{0.25in}
\map{\hat{f}}{\strs\times\strs}{\rats} \]
computable in time $O(n^k)$ and $O(n^j)$, respectively, such that for all $w\in\strs$ and $r\in\nums$,
\[ \left| \hat{d}(w,0^r) - d(w) \right| \le 2^{-r} \hspace{0.25in}\mbox{and}\hspace{0.25in} \left| \hat{f}(w,0^r) - f(0.w) \right| \le 2^{-r}\;. \]

Fix an input $x\in\strs$ and let $n = |x|$.

Fix an $r\in\nums$.  We will choose $m$ to be a sufficiently large integer (depending on $n$ and $r$) to be determined later.  We prove the lemma by describing a procedure (running time polynomial in $m$) to compute a number $v\in\rats$ such that $\left|d_f(x)-v\right| \le 2^{-r}$.


Here is the procedure:
\begin{enumerate}
\item
Compute dyadic rationals $0\le a\le b \le 1$, both with denominator $2^m$, so that $[a,b]$ approximates $\Delta_x$ to within less than $2^{-m}$ for each endpoint:
\begin{enumerate}
\item
Compute $c_0 = \hat{f}(x,0^{m+2})$ and round $c_0$ to the nearest $a\in\dyads$ with denominator $2^m$ so that $|c_0 - a| \le 2^{-(m+1)}$.  Notice that
\[ |f(0.x) - a| \le |f(0.x) - c_0| + |c_0 - a| \le 2^{-(m+2)} + 2^{-(m+1)} < 2^{-m}\;. \]
In other words,
\[ a - 2^{-m} < f(0.x) < a + 2^{-m}\;. \]
(Note that $f(0.x)$ is the left endpoint of $\Delta_x$.)
\item
If $x = 1^n$, then let $b := 1$.  Otherwise, let $x'$ be the lexicographical successor of $x$ in $\two^n$, and compute $c_1 = \hat{f}(x',0^{m+2})$.  Let $b$ be the dyadic rational with denominator $2^m$ closest to $c_1$.  Similarly to $a$, we have
\[ b - 2^{-m} < f(0.x') = f(0.x+2^{-n}) < b + 2^{-m}\;. \]
(Note that $f(0.x+2^{-n})$ is the right endpoint of $\Delta_x$.)
\item
Without loss of generality, we can assume that $0\le a\le b\le 1$: if necessary, reset $a := \min(\max(a,0),1)$ then $b := \min(\max(a,b),1)$.  These adjustments don't affect the inequalities above.
\end{enumerate}
\item
Let $S$ be the set of all $\sqsubseteq$-minimal strings $w$ such that $\Gamma_w \subseteq [a,b]$.
\item
Finally, compute
\[ v := 2^n \sum_{w\in S} 2^{-|w|} \hat{d}(w,0^m)\;. \]
\end{enumerate}

Notice that no string in $S$ is a proper prefix of any other string in $S$; hence the sets $\Gamma_w$ for $w\in S$ are pairwise disjoint except for endpoints.  Further, it is clear that $\bigcup_{w\in S} \Gamma_w = [a,b]$ if $a<b$ (otherwise, $S = \emptyset$).

We claim that $S$, and hence $v$, can be computed in time polynomial in $m$, with a polynomial time bound exponent depending only on $k$ and $j$.  This follows from three facts:
\begin{enumerate}
\item
$a$ and $b$ can be computed in time $O(m^j)$.
\item
Every string in $S$ has length at most $m$.
\item
There are at most two strings in $S$ of any given length.
\end{enumerate}
Fact~1 is clear from the procedure description.  For Fact~2, notice that, since $a$ and $b$ have denominator $2^m$, if $w$ is any string such that $\Gamma_w\subseteq [a,b]$ and $|w|>m$, then removing the last bit of $w$ yields a proper prefix $w'\sqsubset w$ such that $\Gamma_{w'}\subseteq [a,b]$ as well, and so $w$ is not $\sqsubseteq$-minimal, and thus $w\notin S$.

Similarly for Fact~3, if $w_1,w_2,w_3\in S$ are any three distinct strings given in lexicographical order, and $|w_1| = |w_3|$, then $|w_2| < |w_1|$.  To see this, suppose $|w_2| \ge |w_1|$.  Let $w'$ be the result of removing the last bit of $w_2$.  Then $\Gamma_{w'}$ includes $\Gamma_{w_2}$ and another dyadic interval of length $2^{-|w_2|}$ immediately to the left or right of $\Gamma_{w_2}$.  In either case, the left end point of $\Gamma_{w'}$ is not to the left of that of $\Gamma_{w_1}$, and the right endpoint of $\Gamma_{w'}$ is not to the right of that of $\Gamma_{w_3}$.  So $\Gamma_{w'} \subseteq [a,b]$, which means that $w_2$ is not $\sqsubseteq$-minimal, and thus $w_2\notin S$.

Thus $S$ has at most $2m+1$ strings, each of length at most $m$, and so the following greedy algorithm for computing $S$ (given $a$ and $b$) runs in time $O(m^2)$:
\begin{algo}
$S \assn \emptyset$ \\
$z \assn a$ \\
WHILE $z<b$ DO \\
\>Let $w\in\strs$ be shortest such that $z=0.w$ and $z+2^{-|w|}\le b$ \\
\>$S \assn S \union \{w\}$ \\
\>$z \assn z+2^{-|w|}$ \\
END-WHILE \\
return $S$
\end{algo}
It then follows that $v$ can be computed in from $S$ (in Step~3, above) in time $O(m^{k+1})$.

It remains to show that $m$ can be chosen so that $v$ is sufficiently close to $d_f(x)$.  First, note that, due to the closeness of our approximations to the endpoints of $\Delta_x$,
\begin{equation}\label{eqn:sandwich}
d^-(x;m) \le 2^{n-m} \sum_{y\in\two^m\;\colon\; \Gamma_y\subseteq [a,b]} d(y) \le d^+(x;m)\;.
\end{equation}
(The sum in the middle includes all the terms of the sum on the left, and the sum on the right includes all the terms of the sum in the middle.)

Since $[a,b] = \bigcup_{w\in S} \Gamma_w$, and the intervals $\Gamma_w$ intersect only at endpoints, we can rewrite the sum in the middle of (\ref{eqn:sandwich}) as
\[ 2^{n-m} \sum_{w\in S} \left(\sum_{y\in\two^m\;\colon\; w\sqsubseteq y} d(y) \right) = 2^{n-m} \sum_{w\in S} 2^{m-|w|} d(w) = 2^n \sum_{w\in S} 2^{-|w|} d(w)\;, \]
the first equality owing to the fact that $d$ is a martingale.  So Equation~(\ref{eqn:sandwich}) becomes
\begin{equation}\label{eqn:sandwich-again}
d^-(x;m) \le 2^n \sum_{w\in S} 2^{-|w|} d(w) \le d^+(x;m)\;.
\end{equation}

Now we use the fact that $\left|\hat{d}(w,0^m) - d(w)\right| \le 2^{-m}$ for all $w\in S$.  From (\ref{eqn:sandwich-again}) we get
\begin{align*}
d^-(x;m) - 2^{n-m}\sum_{w\in S} 2^{-|w|} &\le 2^n \sum_{w\in S}2^{-|w|}[d(w) - 2^{-m}] \le 2^n \sum_{w\in S}2^{-|w|}\hat{d}(w,0^m) = v \\
&\le 2^n \sum_{w\in S}2^{-|w|}[d(w) + 2^{-m}] \le d^+(x;m) + 2^{n-m}\sum_{w\in S} 2^{-|w|}\;.
\end{align*}
We have $\sum_{w\in S} 2^{-|w|} = b - a \le 1$, so the above inequality implies
\[ d^-(x;m) - 2^{n-m} \le v \le d^+(x;m) + 2^{n-m}\;. \]

Since $d^-(x;m) \le d_f(x) \le d^+(x;m)$ (as follows from Lemma~\ref{lem:monotone}), it is clear then that
\[ \left| d_f(x)  - v \right| \le d^+(x;m) - d^-(x;m) + 2^{n-m} \le 2^{n-m} + 2^{n+1} \left(\frac{3}{4}\right)^m = 2^{n-m} + 2^{n-2m+1} 3^m \]
by Lemma~\ref{lem:squeeze}.  To bound $\left| d_f(x) - v \right|$ above by $2^{-r}$, it suffices that $2^{n-m} \le 2^{-(r+1)}$ and that $2^{n-2m+1} 3^m \le 2^{-(r+1)}$.  That is,
\[ m \ge n+r+1 \hspace{0.25in}\mbox{and}\hspace{0.25in} m \ge \frac{n+r+2}{2 - \log_2 3}\;. \]
So it suffices to set $m := 4(n+r+2) = O(n+r)$.  This makes the entire computation time for $v$ polynomial in $n$ and $r$, and in fact, $v$ can be computed in time $O((n+r)^q)$, where $q := \max(j,2,k+1)$.
\end{proof}

\begin{proof}[Proof of Lemma~\ref{lem:main}]
Let $f$ and $x_0$ be as in Lemma~\ref{lem:main}, and suppose $f$ is weakly $n^j$-computable for some $j$.  If $x_0\in\dyads$, then it is clearly not $n^1$-random, and we are done.  Otherwise, fix $k\in\nums$, and assume that $f(x_0)$ is not $n^k$-random.  Let $\ell = \floor{f(0)}$, let $h = \ceiling{f(1)}$, and let $m\ge 0$ be the least natural number such that $2^m \ge h-\ell$.  For all $x\in [0,1)$, define
\[ g(x) := 2^{-m}(f(x) - \ell)\;, \]
and define $g(1) := 1$.  Then $\map{g}{[0,1]}{[0,1]}$ is monotone ascending, weakly $n^j$-computable on $[0,1]$, and strongly increasing at $x_0$ on $[0,1]$ by Observation~\ref{obs:scale-shift}.  Further, since $f(x_0)$ is not $n^k$-random, it follows from Observation~\ref{obs:scale-shift-random} that $g(x_0) = 2^{-m}(f(x_0)-\ell)$ is not $n^k$-random, either.  Thus by Proposition~\ref{prop:conservative} there exists a $t\in\nums$---depending only on $k$---and a conservative, $n^t$-computable martingale $d$ that succeeds on $g(x_0)$.  By Lemmata~\ref{lem:key} and \ref{lem:p-computable} (letting $0.s$ be $x_0$), the $g$-pullback $d_g$ of $d$ succeeds on $x_0$ and is $n^q$-computable for some $q$ depending only on $j$ and $t$, with the latter depending only on $k$.  Thus $x_0$ is not $n^q$-random.
\end{proof}

To prove Theorem~\ref{thm:main}, we first show that the monotonicity assumption in Corollary~\ref{cor:main} is dispensible.  We do this by tweaking a nonmonotone function into a monotone one with the same desirable properties.

\begin{lemma}\label{lem:discharge-monotone}
Let $\map{f}{[0,1]}{\reals}$ be weakly p-computable on $[0,1]$.  Suppose that there exists $x_0\in [0,1]$ such that $f$ strongly increases at $x_0$ on $[0,1]$.  Then there exists a monotone ascending function $\map{g}{[0,1]}{\reals}$ that is weakly p-computable on $[0,1]$, strongly increases at $x_0$ on $[0,1]$, and satisfies $g(x_0) = f(x_0)$.
\end{lemma}

\begin{proof}[Proof of Lemma~\ref{lem:discharge-monotone}]
We first define $g$ on $\dyads\intersection [0,1]$ to be monotone.  Extending $g$ to domain $[0,1]$ will then be trivial.

The idea is that we give priority to dyadic rationals with smaller denominators, and for any point $x\in\dyads$, we let $g(x) := f(x)$ unless this violates monotonicity with a neighboring point of higher priority (i.e., lower denominator).  If so, we adjust $g(x)$ just enough to avoid the violation.

Here we give a recursive definition of $g$ restricted to $\dyads\intersection [0,1]$ based on $f$.  For any $q\in \dyads\intersection (0,1)$, let
$y_q\in\strs$ be the unique string such that $q = 0.y_q1$.  We define $e_q := |y_q|+1$ and call this the \emph{exponent} of $q$.  By convention, the exponents $e_0$ of $0$ and $e_1$ of $1$ are both $0$.  Define
\[ q^- := 0.y_q\;, \]
and define
\[ q^+ := \left\{ \begin{array}{ll}
1 & \mbox{if $y_q\in\{1\}^*$\;,} \\
0.z1 & \mbox{if $(\exists z\in\strs)(\exists w\in\{1\}^*)[y_q = z0w]$\;.}
\end{array} \right. \]
(Note that $z$ and $w$ are unique if they exist.)  The points $q^-$ and $q^+$ are the dyadic rationals closest to $q$ (on the left and right side, respectively) whose exponents are less than that of $q$.

Now we define $g(0) := f(0)$, $g(1) := f(1)$, and for each $q\in \dyads \intersection (0,1)$,
\begin{equation}\label{eqn:g-recursive}
g(q) := \max(g(q^-),\min(g(q^+),f(q)))\;.
\end{equation}
The recursion is well-founded because $q^-$ and $q^+$ have smaller exponents than $q$.

\begin{claim}\label{clm:monotone}
The function $g$ is monotone ascending on $\dyads \intersection [0,1]$.
\end{claim}

\begin{proof}[Proof of Claim~\ref{clm:monotone}]
Let $p$ and $q$ be dyadic rationals with $0\le p < q\le 1$.  We proceed by induction on $e := \max(e_p,e_q)$ to show that $g(p) \le g(q)$.  If $e = 0$, then we have $p=0$ and $q=1$, and clearly, $g(0) = f(0) < f(1) = g(1)$ by the constraints on $f$.  If $e>0$, we have three cases:
\begin{enumerate}
\item
If $e_p < e_q$, then we have $p \le q^-$ by the maximality of $q^-$, and so by the inductive hypothesis, $g(p) \le g(q^-)$.  Then by the recursive definition of $g(q)$, we have $g(q^-) \le g(q)$, hence $g(p) \le g(q)$.
\item
If $e_p > e_q$, then $p^+ \le q$, and so by the inductive hypothesis, $g(p^-) \le g(p^+) \le g(q)$.  By the recursive definition of $g(p)$ (and the fact that $g(p^-) \le g(p^+)$), we have $g(p) \le g(p^+)$, whence $g(p) \le g(q)$.
\item
If $e_p = e_q > 0$, then $|y_p| = |y_q|$.  Let $y$ be the longest common prefix of $y_p$ and $y_q$.  Then clearly, $y0\sqsubseteq y_p$ and $y1\sqsubseteq y_q$.  Let $r = 0.y1$.  Since $y$ is shorter than $y_p$ and $y_q$, we have $e_r < e_p$ and $e_r < e_q$, and in addition, $p < r < q$.  Thus $p^+ \le r\le q^-$, and so by the inductive hypothesis, $g(p^+) \le g(r) \le g(q^-)$.  By an argument similar to the other two cases, we have $g(p) \le g(p^+)$ and $g(q^-) \le g(q)$.  Thus $g(p) \le g(q)$.
\end{enumerate}
This ends the proof of Claim~\ref{clm:monotone}.
\end{proof}

\begin{claim}\label{clm:p-computable}
The function $g$ is p-computable on $\dyads \intersection [0,1]$.
\end{claim}

\begin{proof}[Proof of Claim~\ref{clm:p-computable}]
First $g(0) = f(0)$ and $g(1) = f(1)$, so $g$ is p-computable at $0$ and at $1$.  For any $q \in \dyads \intersection (0,1)$, we have
\[ g(q) = \max(g(q^-), \min(g(q^+),f(q)))\;. \]
First notice that for any $r\in\nums$, if $a$, $b$, and $c$ are such that $|a - g(q^-)| \le 2^{-r}$, $|b - g(q^+)| \le 2^{-r}$, and $|c-f(q)| \le 2^{-r}$ then it is not hard to see that
\[ \left| \max(a,\min(b,c)) - g(q) \right| \le 2^{-r}\;. \]
Thus to approximate $g(q)$ to within $2^{-r}$, it suffices to approximate $g(q^-)$, $g(q^+)$, and $f(q)$ each to within $2^{-r}$.

Let $q = 0.y_q1$ where $y_q$ is as above.  Unwinding the recursion of Equation~\ref{eqn:g-recursive}, it becomes apparent that $g(q)$ only depends on $f$ at $1$ and at points of the form $0.y$ and $0.y1$ for $y\sqsubseteq y_q$.  So to approximate $g(q)$ we only need to approximate $f$ on these points.  Here is a nonrecursive polynomial-time algorithm, equivalent to Equation~\ref{eqn:g-recursive}, to approximate $g$ on $\dyads\intersection [0,1)$.  It assumes a p-approximator $\hat{f}$ for $f$ on $[0,1)$ and a p-approximator $\hat{f}_1$ for $f(1)$.
\begin{algo}
Algorithm for $\hat{g}(x,0^r)$ \\
// $x\in\strs$ and $r\in\nums$ \\
// Outputs a $y\in\rats$ such that $|y - g(0.x)| \le 2^{-r}$ \\
\>Remove any trailing zeros from $x$ \\
\>$\ell \assn \hat{f}(\emptystr,0^r)$ \\
\>$h \assn \hat{f}_1(0^r)$ \\
\>$s \assn \emptystr$ \\
\>FOR $i\assn 0$ TO $|x|-1$ DO \\
\>\>$b \assn x[i]$ \\
\>\>IF $b = 0$ THEN \\
\>\>\>$h \assn \max(\ell,\min(h,\hat{f}(s1,0^r)))$ \\
\>\>ELSE // IF $b = 1$ THEN \\
\>\>\>$\ell \assn \max(\ell,\min(h,\hat{f}(s1,0^r)))$ \\
\>\>$s \assn sb$ \\
\>END-FOR \\
\>OUTPUT $\ell$ and STOP
\end{algo}

The algorithm above clearly runs in polynomial time.  The proof that it correctly approximates $g(0.x)$ uses the following key loop invariant: At the start of each iteration of the FOR-loop, $s\sqsubset x$, and in addition,
\[ |\ell - g((0.s1)^-)| \le 2^{-r} \; \mbox{ and } \; |h - g((0.s1)^+)| \le 2^{-r}\;. \]
We omit the details.  This ends the proof of Claim~\ref{clm:p-computable}.
\end{proof}

\begin{claim}\label{clm:g-strongly-increasing}
There exists a $C>0$ such that for all $x\in [0,1]\intersection\dyads - \{x_0\}$,
\begin{equation}\label{eqn:g-strongly-increasing}
\frac{g(x) - f(x_0)}{x-x_0} \ge C\;.
\end{equation}
\end{claim}

\begin{proof}[Proof of Claim~\ref{clm:g-strongly-increasing}]
We can let $C$ be any constant witnessing the strong increase of $f$ at $x_0$ on $[0,1]$.  We proceed by induction on the exponent $e_x$ of $x$.  This is clear when $e_x = 0$.  If $e_x > 0$, then
\[ g(x) = \max(g(x^-),\min(g(x^+),f(x))) \]
by Equation~\ref{eqn:g-recursive}.  If $g(x) = f(x)$, then were are clearly done.  Suppose $g(x) < f(x)$.  Then Equation~\ref{eqn:g-strongly-increasing} is still satisfied if $x<x_0$, so suppose that $x>x_0$.  We have $g(x) = g(x^+)$, and so, using the inductive hypothesis,
\[ \frac{g(x) - f(x_0)}{x-x_0} = \frac{g(x^+) - f(x_0)}{x-x_0} \ge \frac{g(x^+) - f(x_0)}{x^+ - x_0} \ge C\;. \]
A similar argument using $g(x^-)$ works if $g(x) > f(x)$.  This ends the proof of Claim~\ref{clm:g-strongly-increasing}.
\end{proof}

We now extend the definition of $g$ to all of $[0,1]$ by
\[ g(x) := \sup\{g(y) \mid y\in \dyads\intersection[0,x]\}\;, \]
except that we define $g(x_0) := f(x_0)$.  (If $x_0\in\dyads$, then we already have $g(x_0) = f(x_0)$, because $g(x_0^-) < f(x_0) < g(x_0^+)$ by Equation~\ref{eqn:g-strongly-increasing}.)  The claims imply that $g$ has all the requisite properties.
\end{proof}

\begin{proof}[Proof of Theorem~\ref{thm:main}]
Let $I$, $f$, and $r$ be as in the statement of the theorem.  We can assume that $f$ strongly increases at $x$, for otherwise we apply the foregoing argument to $-f$, using Observations~\ref{obs:scale-shift-random} and \ref{obs:scale-shift} to get that $f(r)$ is p-random.  We can choose some dyadic interval $\Gamma_w = [0.w,0.w+2^{-|w|}]\subseteq I$ containing $r$ on which $f$ is weakly p-computable and strongly increases at $x$.  For all $x\in [0,1]$, define
\[ g(x) := f(0.w + 2^{-|w|}x)\;. \]
By Observation~\ref{obs:scale-shift}, $g$ is weakly p-computable on $[0,1]$ and strongly increases at the point $s := 2^{|w|}(r-0.w)$ on $[0,1]$.  By Lemma~\ref{lem:discharge-monotone}, there is a monotone ascending function $h$ that is weakly p-computable on $[0,1]$, is strongly increasing at $s$ on $[0,1]$, and satisfies $h(s) = g(s)$.  By Observation~\ref{obs:scale-shift-random}, $s$ is p-random.  By Corollary~\ref{cor:main}, $h(s)$ is p-random, and clearly, $h(s) = g(s) = f(r)$, which proves the theorem.
\end{proof}

\section{Some p-randomness-preserving functions}
\label{sec:apps}

Here is the class of functions we will consider:

\begin{definition}\rm
Let $I\subseteq\reals$ be an open interval.  A function $\map{f}{I}{\reals}$ is \emph{well-behaved on $I$} if $f$ is locally weakly p-computable and strongly varying at each of the p-random points in $I$.
\end{definition}

Theorem~\ref{thm:main} gives us the following corollary:

\begin{corollary}
If a function $f$ is well-behaved on an interval $I$, then $f$ preserves p-randomness, i.e., $f$ maps p-random points in $I$ to p-random points.
\end{corollary}

A wide variety of functions are well-behaved and hence preserve p-random\-ness, including addition and multiplication by nonzero p-computable numbers, nonconstant polynomial and rational functions with p-computable coefficients, and all the familiar transcendental functions---exponential, logarithmic, trigonometric, etc.  (Define a function to be $0$ where it would otherwise be undefined.)  Although these functions may not be strongly varying at all points, they are strongly varying at all p-random points.


\begin{definition}\rm
A sequence $c_0,c_1,c_2,\ldots \in \reals$ is \emph{uniformly p-computable} if there exists a polynomial-time function $\map{\hat{c}}{\strs\times\strs}{\rats}$ such that for all $n,r\in\nums$,
\[ \left| \hat{c}(0^n,0^r) - c_n \right| \le 2^{-r}\;. \]
\end{definition}

\begin{definition}\rm
Let $I\subseteq\reals$ be an open interval.  We say that a function $\map{f}{I}{\reals}$ is \emph{p-analytic on $I$} if there exists a p-computable point $x_0\in I$ and a uniformly p-computable sequence $c_0,c_1,c_2,\ldots$ such that for all $x\in I$,
\[ f(x) = \sum_{n=0}^{\infty} c_n(x-x_0)^n\;, \]
and the power series on the right converges absolutely for all $x\in I$.
\end{definition}

Note that if $f$ is p-analytic on $I$, then $f$ is $C^1$ on $I$.  In this section we prove the following theorem:

\begin{theorem}\label{thm:p-analytic}
Let $I\subseteq\reals$ be an open interval.  If $\map{f}{I}{\reals}$ is nonconstant and p-analytic on $I$, then $f$ is well-behaved on $I$.
\end{theorem}

Theorem~\ref{thm:p-analytic} follows from the two lemmas below:

\begin{lemma}\label{lem:p-analytic}
Let $J\subseteq \reals$ be an open interval and let $I$ be a dyadic interval such that $I\subseteq J$.  If $f$ is p-analytic on $J$, then $f$ is weakly p-computable on $I$.
\end{lemma}

\begin{proof}[Proof of Lemma~\ref{lem:p-analytic}]
Our proof mirrors standard results from calculus.  Let $c_0,c_1,c_2,\ldots$ be a uniformly p-computable sequence witnessed by $\hat{c}$, and let $x_0\in J$ be such that $f(x) = \sum_{n=0}^{\infty} c_n(x-x_0)^n$, with the right-hand side converging absolutely, for all $x\in J$.  Let $r = \sup\{|x-x_0| \mathrel{:} x\in I\}$, and let $x\in I$ with $|x-x_0| = r$.  Since $x\in J$ and $J$ is open, there must be an $\eps > 0$ such that $\sum_{n=0}^{\infty} |c_n|(r+\eps)^n < \infty$.  Hence all the terms $|c_n|(r+\eps)^n$ are upper bounded by some constant $C\ge 1$ independent of $n$.  This implies in turn that for all $-r\le z \le r$ and $m \ge 0$, we can bound the tail of the series:
\begin{equation}\label{eqn:tail-bound}
\left|\sum_{n=m}^{\infty} c_n z^n\right| \le \sum_{n=m}^{\infty} |c_n| r^n \le \sum_{n=m}^\infty C\left(\frac{r}{r+\eps}\right)^n = C\left(\frac{r}{r+\eps}\right)^m \frac{r+\eps}{\eps} \le 2^{k-m/\ell}\;,
\end{equation}
where $k = \ceiling{\lg(C(r+\eps)/\eps)}$ and $\ell := \ceiling{1/\lg((r+\eps)/r)}$.

Let $\hat{x}$ be a p-approximator for $x_0$.  Fix $w\in\two^*$ such that $I = \Gamma_w$.  For $a\in\two^*$ and $s\in\nums$, define
\begin{align*}
m_s &:= \ell(s+k+1)\;, \\
b_s &:= \bigceiling{\lg\left(2+\max_{n<m_s}\left\{|\hat{c}(0^n,\emptystr)|,|0.wa-\hat{x}(\lambda)|\right\}\right)}\;, \\
\hat{f}(a,0^s) &:= \sum_{n=0}^{m_s-1} \hat{c}(0^n,0^{s+b_s n + 2m_s + 1})\left[0.wa-\hat{x}(0^{s+b_s n + 2m_s + 1})\right]^n\;.
\end{align*}
Clearly, $\hat{f}$ is polynomial-time computable.   We then have, for all $a\in\two^*$ and $s\in\nums$, letting $e(n,s)$ denote $s+b_s n + 2m_s + 1$,
\begin{align*}
&\bigabs{\hat{f}(a,0^s) - f(0.wa)} \\
&= \bigabs{\sum_{n=0}^{m_s-1} \hat{c}(0^n,0^{e(n,s)})\left[0.wa-\hat{x}(0^{e(n,s)})\right]^n - \sum_{n=0}^{\infty} c_n(0.wa-x_0)^n} \\
&\le \bigabs{\sum_{n=0}^{m_s-1}\left[\hat{c}(0^n,0^{e(n,s)})\left[0.wa-\hat{x}(0^{e(n,s)})\right]^n - c_n(0.wa-x_0)^n\right]} \\
&\mbox{}\hspace{0.5in} + \bigabs{\sum_{n=m_s}^\infty c_n (0.wa - x_0)^n} \\
&\le \sum_{n=0}^{m_s-1} \bigabs{\hat{c}(0^n,0^{e(n,s)})\left[0.wa-\hat{x}(0^{e(n,s)})\right]^n - c_n(0.wa-x_0)^n} \\
&\mbox{}\hspace{0.5in} + 2^{k-m_s/\ell}\;,
\end{align*}
by Equation~(\ref{eqn:tail-bound}) because $|0.wa - x_0| \le r$.  By our choice of $m_s$, we have $2^{k-m_s/\ell} = 2^{-s-1}$, which bounds the tail term.  For the term being summed, we use the formula for the difference of two products as a telescoping sum:
\[ \alpha_1\cdots\alpha_n - \beta_1\cdots \beta_n = \sum_{i=1}^n \alpha_1\cdots\alpha_{i-1}(\alpha_i - \beta_i)\beta_{i+1}\cdots\beta_n\;. \]
Our choice of $b_s$ ensures that
\[ 2^{b_s} \ge \max_{n<m_s}\left\{|\hat{c}(0^n,0^{e(n,s)})|,|c_n|,|0.wa-\hat{x}(0^{e(n,s)})|,|0.wa-x_0|\right\}\;. \]
Combining these gives
\begin{align*}
&\bigabs{\hat{c}(0^n,0^{e(n,s)})\left[0.wa-\hat{x}(0^{e(n,s)})\right]^n - c_n(0.wa-x_0)^n} \\
&\le 2^{b_s n}\left(\bigabs{\hat{c}(0^n,0^{e(n,s)}) - c_n} + n\bigabs{x_0 - \hat{x}(0^{e(n,s)})}\right) \\
&\le 2^{b_s n}(n+1)2^{-e(n,s)} = (n+1)2^{-s-2m_s-1}
\end{align*}
for all $n<m_s$.  Thus
\[ \bigabs{\hat{f}(a,0^s) - f(0.wa)} \le (m_s)^2 2^{-2m_s} 2^{-s-1} + 2^{-s-1} \le 2^{-s}\;, \]
and so $\hat{f}(a,0^s)$ approximates $f(0.wa)$ closely enough.
\end{proof}

\begin{lemma}\label{lem:root-of-f}
Suppose $f$ is p-analytic and nonconstant in some open interval $I$.  If $r\in I$ satisfies $f(r)=0$, then $r$ is p-computable.
\end{lemma}

\begin{proof}[Proof sketch of Lemma~\ref{lem:root-of-f}]
Let $f(x) = \sum_{n=0}^\infty c_n(x - x_0)^n$, where $x_0$ is p-computable, the $c_n$ are uniformly p-computable, and the sum converges absolutely on $I$.  Let $r$ be such that $f(r) = 0$.  Expressing $f(x)$ as a power series about $r$ gives $f(x) = \sum_{n=1}^\infty c_n'(x-r)^n$ for some constants $c_n'$.  Since $f$ is nonconstant, there is a least $m>0$ such that $c_m' \ne 0$.  Then $f(r) = f'(r) = f''(r) = \cdots =f^{(n-1)}(r) = 0$, but $f^{(n)}(r) \ne 0$.

It is easy to observe that if a function $g$ is p-analytic on $I$, then so is its derivative $g'$.  Letting $g := f^{(n-1)}$, we see that: (i) $g$ is p-analytic and thus weakly p-computable; (ii) $g(r) = 0$; and (iii) $g'(r) \ne 0$.  Hence there is a neighborhood $N$ of $r$ such that $g(x)$ changes sign at $x=r$ and nowhere else.  This allows us to find $r$ quickly using binary search, testing the sign of $g(x)$ for various $x\in N$.
\end{proof}

\begin{proof}[Proof of Theorem~\ref{thm:p-analytic}]
We know already that, since $f$ has a continuous derivative, it strongly varies at any point $r$ such that $f'(r) \ne 0$ (hence if $r$ is p-random then so is $f(r)$).  If $f'(r) = 0$, then $r$ is p-computable by Lemma~\ref{lem:root-of-f}, and thus not p-random.
\end{proof}

\begin{corollary}
Let $r$ be p-random.  Then so are $e^r$, $\sin r$, $\cos r$, and $\tan r$.  If $r>0$, then $\ln r$ is p-random.  If $f$ is any fixed rational function whose numerator and denominator have p-computable coefficients, and $f$ is defined at $r$, then $f(r)$ is p-random.  If $c\ne 0$ is p-computable, then $cr$ and $c+r$ are p-random.
\end{corollary}

\begin{proof}
All these functions are p-analytic in some neighborhood of any point in their domains.
\end{proof}

\section{The tightness of Theorem~\ref{thm:main}}
\label{sec:tight}

In this section, we give evidence that the strongly varying property of $f$ in Theorem~\ref{thm:main} is essentially tight. To this end, we concoct monotone functions that deviate only slightly from strongly varying, but none of whose outputs are p-random.  For example, one could have $f(0.\sigma) = 0.\tau$, where the sequence $\tau$ results from the sequence $\sigma$ by inserting zeros into $\sigma$ very sparsely but infinitely often, in places that are easy for a martingale to find and bet on.

\begin{definition}\rm
Fix $Z\subseteq\nums$, and define its \emph{census function} $c(i) := |Z\intersection\{0,\ldots,i\}|$ for all $i\in\nums$.
\begin{enumerate}
\item
For every $s\in\seqs$, define $s_Z \in\seqs$ such that, for all $i\in\nums$,
\[ s_Z[i] = \left\{ \begin{array}{ll}
s[i-c(i)] & \mbox{if $i\notin Z$,} \\
0 & \mbox{if $i\in Z$.}
\end{array} \right. \]
\item
Let $\map{f_Z}{[0,1)}{\reals}$ be the function mapping $0.s$ to $0.(s_Z)$ for every $s\in\seqs$ with infinitely many zeros.
\end{enumerate}
\end{definition}

Note that $s_Z$ results from $s$ by inserting zeros at the positions $i\in Z$, shifting bits of $s$ to the right to make room.

\begin{observation}\label{obs:tight}
Let $Z\subseteq\nums$ be arbitrary, and let $c$ be its census function.
\begin{enumerate}
\item
For any $s\in\seqs$ with infinitely many zeros,
\begin{equation}\label{eqn:fZ}
f_Z(0.s) = 0.(s_Z) = \sum_{i=0}^\infty s[i]2^{-(i+c(i)+1)}\;.
\end{equation}
\item\label{item:fZ-monotone}
The function $f_Z$ is monotone ascending, and if $\nums - Z$ is infinite, then $f_Z$ is one-to-one.
\item\label{item:fZ-computable}
If the predicate, ``$n\in Z$'' is computable in time polynomial in $n$, then $f_Z$ is weakly p-computable.
\item\label{item:tight}
If $Z$ is infinite and the predicate, ``$n\in Z$'' is computable in time polynomial in $n$, then $f(x)$ is never p-random for any $x\in [0,1)$.
\end{enumerate}
\end{observation}

\begin{proof}[Proof sketch]
Equation~(\ref{eqn:fZ}) is a routine application of the definition of $s_Z$.  Point~(\ref{item:fZ-monotone}.) is obvious.  For Point~(\ref{item:fZ-computable}.), note that the list $\tuple{c(0),c(1),\ldots,c(n)}$ is computable in time polynomial in $n$, which makes the sum in Equation~(\ref{eqn:fZ}) easy to approximate to polynomially many terms.  For Point~(\ref{item:tight}.), consider a martingale that bets on a string $w$ iff $|w| \in Z$, in which case it puts all its money on the next bit being $0$.  This martingale will succeed on any sequence of the form $s_Z$.
\end{proof}

\begin{theorem}\label{thm:tight}
Let $Z\in\nums$ be arbitrary, and let $c$ be its census function.  For any real $x$ and $y$ such that $0\le x < y < 1$,
\begin{equation}\label{eqn:fZ-ratio}
\frac{f_Z(y) - f_Z(x)}{y-x} > 2^{-c(\ceiling{-\lg(y-x)})-1}\;.
\end{equation}
\end{theorem}

If $Z$ is finite, then its census function $c$ is bounded from above, whence Theorem~\ref{thm:tight} says that $f_Z$ is strongly increasing everywhere.  The strength of Theorem~\ref{thm:tight} comes when $Z$ is infinite but extremely sparse, e.g., $Z$ is the range of the one-argument Ackermann function.  Then the theorem implies that $f_Z$ comes very close to being strongly increasing, because the function $c$ grows very slowly.  If, in addition, $Z$ satisfies Observation~\ref{obs:tight}(\ref{item:tight}.), then we get a weakly p-computable, monotone function $f_Z$ that is extremely close to being strongly increasing everywhere, but none of whose outputs is p-random.

\begin{proof}[Proof of Theorem~\ref{thm:tight}]
We first consider the case where $n$ is a positive integer and $y = x + 2^{-n}$.  In this case, we prove that
\begin{equation}\label{eqn:fZ-strong-ratio}
f_Z(x+2^{-n}) - f_Z(x) \ge 2^{-c(n)-n}\;.
\end{equation}
Once Equation~(\ref{eqn:fZ-strong-ratio}) is established, Equation~(\ref{eqn:fZ-ratio}) follows easily by the monotonicity of $f_Z$: setting $n := \ceiling{-\lg(y-x)}$ and noting that $x+ 2^{-n} \le y < x + 2^{1-n}$, we have
\[ \frac{f_Z(y) - f_Z(x)}{y-x} > 2^{n-1}[f_Z(y) - f_Z(x)] \ge 2^{n-1}[f_Z(x+2^{-n}) - f_Z(x)] \ge 2^{-c(n)-1}\;. \]

To establish Equation~(\ref{eqn:fZ-strong-ratio}), we let $s\in\seqs$ be such that $x = 0.s$ (and $s$ has infinitely many zeros).  Similarly, let $x + 2^{-n} = 0.t$ for some $t\in\seqs$ with infinitely many zeros.  It is not too hard to see that $t$ results from $s$ by adding $1$ to $s$ in the $(n-1)$th position, then carrying $1$'s to the left until a zero is reached: Let $k\in\nums$ be largest such that $k<n$ and $s[k] = 0$.  Such a $k$ must exist because $x+2^{-n} < 1$ by assumption.  Then $s$ and $t$ differ only in positions $k$ through $n-1$, where
\begin{align*}
s[k\ldots(n-1)] &= 011\cdots 1, \\
t[k\ldots(n-1)] &= 100\cdots 0.
\end{align*}
Using Equation~(\ref{eqn:fZ})---and noting that $c$ is monotone ascending---we then get
\begin{align*}
f_Z(x+2^{-n}) - f_Z(x) &= f_Z(0.t) - f_Z(0.s) = \sum_{i=1}^\infty t[i]2^{-(i+c(i)+1)} - \sum_{i=1}^\infty s[i]2^{-(i+c(i)+1)} \\
&= \sum_{i=k}^{n-1} (t[i] - s[i])2^{-(i+c(i)+1)} = 2^{-(k+c(k)+1)} - \sum_{i=k+1}^{n-1}2^{-(i+c(i)+1)} \\
&\ge 2^{-(k+c(k)+1)} - \sum_{i=k+1}^{n-1}2^{-(i+c(k)+1)} = 2^{-c(k)}\left(2^{-k-1} - \sum_{i=k+1}^{n-1}2^{-i-1}\right) \\
&= 2^{-c(k)-n} \ge 2^{-c(n)-n}\;,
\end{align*}
which establishes Equation~(\ref{eqn:fZ-strong-ratio}).
\end{proof}

\section{P-Measure}
\label{sec:measure}

There is a close connection between resource-bounded measure and resource-bounded randomness, and so it stands to reason that our results about the latter have some bearing on the former.  This is indeed the case, at least with regard to p-measure and p-randomness, as given in Theorem~\ref{thm:measure}, below.

We start with what is now the standard definition of ``p-measure $0$'' and some basic facts related to it.  See, for example, Lutz~\cite{Lutz:measure,Lutz:measure-survey} and Ambos-Spies, Terwijn, \& Zheng~\cite{ATZ:weak-complete}.

\begin{definition}\rm
A set $X\subseteq\seqs$ has \emph{p-measure $0$} (written $\mu_\textup{p}(X) = 0$) iff there exists a p-computable martingale $d$ such that, for all $s\in X$, \ $d$ succeeds on $s$.

A set $X\subseteq [0,1]$ has \emph{p-measure $0$} iff $\{x\in\seqs \mid 0.x \in X\}$ has p-measure $0$ in the sense above.
\end{definition}

\begin{observation}
A sequence $s\in\seqs$ is p-random if and only if $\{s\}$ does not have p-measure $0$.
\end{observation}

The next proposition follows from the fact that, for every $k\in\nums$, there exists $\ell\in\nums$ and an $n^\ell$-computable martingale $d$ that succeeds on all non-$n^k$-random sequences \cite{Lutz:measure,ATZ:weak-complete}.

\begin{proposition}\label{prop:measure-vs-random}
A set $X\subseteq\seqs$ has p-measure $0$ if and only if there exists $k\in\nums$ such that $X$ contains no $n^k$-random sequences.  The same holds \textit{mutatis mutandis} for $X\subseteq [0,1]$.
\end{proposition}

The next theorem follows immediately from Lemma~\ref{lem:main} and Proposition~\ref{prop:measure-vs-random}, and it implies Corollary~\ref{cor:main}.

\begin{theorem}\label{thm:measure}
Let $X\subseteq [0,1]$ have p-measure $0$.  Suppose $\map{f}{[0,1]}{[0,1]}$ is weakly p-computable, is monotone ascending on $[0,1]$, and strongly increases at each $x \in f^{-1}(X)$.  Then $f^{-1}(X)$ has p-measure $0$.
\end{theorem}

\begin{proof}
By Proposition~\ref{prop:measure-vs-random}, there exists $k\in\nums$ such that no $x\in X$ is $n^k$-random.  Suppose that $f$ is weakly $n^j$-computable, for some $j$.  Then by Lemma~\ref{lem:main}, there exists $q$ such that no $x_0 \in f^{-1}(X)$ is $n^q$-random.  So again by Proposition~\ref{prop:measure-vs-random}, $f^{-1}(X)$ has p-measure $0$.
\end{proof}

It is interesting to note that in Theorem~\ref{thm:measure} we only require $f$ to strongly increase on \emph{some} neighborhood of each point $x \in f^{-1}(X)$.  We require no uniform choice of constant $C$ in Definition~\ref{def:strongly-varies}.  So for example, the theorem applies to functions such as
\[ f(x) := \left\{ \begin{array}{ll}
(e^{1-1/x}+1)/2 & \mbox{if $x>0$,} \\
0 & \mbox{if $x=0$,}
\end{array} \right. \]
which strongly increases at all points in $[0,1]$ but for which no single constant $C$ suffices.

%
%
%
%
%
%

\section{Further research}
\label{sec:open}

P-randomness-preserving functions are clearly closed under composition.  Are well-behaved functions closed this way also?

Theorem~\ref{thm:tight} notwithstanding, we are at a loss to prove a converse to Theorem~\ref{thm:main}.  Is there even a partial converse?  For example, consider the following conjecture about monotone functions:

\begin{conjecture}
If $f$ is weakly p-computable and monotone in a neighborhood of $r\in\reals$ but is not strongly varying at $r$, then $f(r)$ is not p-random.
\end{conjecture}

Theorem~\ref{thm:tight} falls short of proving this conjecture, because it assumes that the set $Z$ is easy to compute.  In general, however, if $f$ is not strongly varying, then the violations to strong variation may come in places that are difficult to detect by a martingale.



\bigskip

The current work may have some connections with previous work of Breutzmann \& Lutz~\cite{BL:nonuniform}, who are chiefly concerned with the resource-bounded measure of complexity classes under certain \emph{nonuniform} measures.

\begin{definition}\rm
A \emph{probability measure} on $\seqs$ is a function $\map{\nu}{\strs}{[0,1]}$ such that $\nu(\emptystr) = 1$, and for all $w\in\strs$,
\[ \nu(w) = \nu(w0) + \nu(w1)\;. \]
\end{definition}

Breutzmann \& Lutz generally show that probability measures that are sufficiently similar give rise to the same notion of resource-bounded measure $0$, at least among complexity classes with weak closure properties.  The following two definitions and propositions suggest that there may be a link between our work and theirs:

\begin{definition}\rm
Let $\nu$ be a probability measure on $\seqs$.  We let the \emph{cumulative function} of $\nu$ be the map $\map{f_\nu}{[0,1]}{[0,1]}$ defined as follows: For any $x\in [0,1]$,
\[ f_\nu(x) := \lim_{n\rightarrow\infty} \sum_{y\in\two^n\;:\;0.y < x} \nu(y)\;. \]
Note that $f_\nu$ is monotone ascending and that $f_\nu(0) = 0$ and $f_\nu(1) = 1$.
\end{definition}

\begin{definition}\rm
Let $\map{f}{[0,1]}{[0,1]}$ be any monotone ascending function such that $f(0) = 0$ and $f(1) = 1$.  Define the \emph{differential probability measure} of $f$ to be the map $\map{\nu_f}{\strs}{[0,1]}$ such that, for all $w\in\strs$,
\[ \nu_f(w) := f(0.w + 2^{|w|}) - f(0.w)\;. \]
Note that $\nu_f$ is a probability measure on $\seqs$.
\end{definition}

\begin{proposition}
For any probability measure $\nu$ and monotone ascending function $\map{f_\nu}{[0,1]}{[0,1]}$ such that $f(0) = 0$ and $f(1) = 1$,
\[ \nu = \nu_f \iff f = f_\nu\;. \]
\end{proposition}

\begin{proposition}
Let $\nu$ be a probability measure on $\seqs$, and let $f = f_\nu$ be its cumulative function.  Then $\nu$ is p-computable if and only if $f$ is weakly p-computable.
\end{proposition}

It would be interesting to pursue these connections further to see if our ideas can provide improvements to their results.

\subsubsection*{Acknowledgments}
I would like to thank Jack Lutz for suggesting this problem and for many interesting and valuable discussions.  The presentation was also helped significantly from an earlier draft by comments from anonymous referees.  I would also like to thank Lance Fortnow for his guidance at a crucial point.

\bibliography{/Users/steve/research/bib/master.bib}

\begin{thebibliography}{10}

\bibitem{AM:randomness-survey}
K.~Ambos-Spies and E.~Mayordomo.
\newblock Resource-bounded measure and randomness.
\newblock In {\em Complexity, Logic and Recursion Theory}, Lecture Notes in
  Pure and Applied Mathematics, pages 1--47. 1997.

\bibitem{ATZ:weak-complete}
K.~Ambos-Spies, S.~A. Terwijn, and X.~Zheng.
\newblock Resource-bounded randomness and weakly complete problems.
\newblock {\em Theoretical Computer Science}, 172:195--207, 1997.

\bibitem{BL:nonuniform}
J.~M. Breutzmann and J.~H. Lutz.
\newblock Equivalence of measures of complexity classes.
\newblock {\em SIAM Journal on Computing}, 29:302--326, 2000.

\bibitem{DLN:arithmetic}
D.~Doty, J.~H. Lutz, and S.~Nandakumar.
\newblock Finite-state dimension and real arithmetic.
\newblock {\em Information and Computation}, 205:1640--1651, 2007.

\bibitem{DHNT:randomness}
R.~Downey, D.~R. Hirschfeldt, A.~Nies, and S.~A. Terwijn.
\newblock Calibrating randomness.
\newblock {\em Bulletin of Symbolic Logic}, 12(3):411--491, 2006.

\bibitem{DH:randomness}
R.~G. Downey and D.~Hirschfeldt.
\newblock {\em Algorithmic Randomness and Complexity}.
\newblock Springer-Verlag, 2010.

\bibitem{Fenner:p-rand-func}
S.~A. Fenner.
\newblock Functions that preserve p-randomness.
\newblock In {\em Proceedings of the 18th International Symposium on
  Fundamentals of Computation Theory}, volume 6914 of {\em Lecture Notes in
  Computer Science}, pages 336--347, 2011.

\bibitem{Lutz:thesis}
J.~H. Lutz.
\newblock Category and measure in complexity classes.
\newblock {\em SIAM Journal on Computing}, 19(6):1100--1131, December 1990.

\bibitem{Lutz:measure}
J.~H. Lutz.
\newblock Almost everywhere high nonuniform complexity.
\newblock {\em Journal of Computer and System Sciences}, 44:220--258, 1992.

\bibitem{Lutz:measure-survey}
J.~H. Lutz.
\newblock The quantitative structure of exponential time.
\newblock In L.~A. Hemaspaandra and A.~L. Selman, editors, {\em Complexity
  Theory Retrospective {II}}, pages 225--260. Springer-Verlag, 1997.

\bibitem{LM:fractals}
J.~H. Lutz and E.~Mayordomo.
\newblock Dimensions of points in self-similar fractals.
\newblock {\em SIAM Journal on Computing}, 38:1080--1112, 2008.

\bibitem{LW:connectivity}
J.~H. Lutz and K.~Weihrauch.
\newblock Connectivity properties of dimension level sets.
\newblock {\em Mathematical Logic Quarterly}, 54:483--491, 2008.

\bibitem{MartinLoef:random}
P.~Martin-L{\"{o}}f.
\newblock The definition of infinite random sequences.
\newblock {\em Information and Control}, 9:602--619, 1966.

\bibitem{Wall:normal}
D.~D. Wall.
\newblock {\em Normal Numbers}.
\newblock PhD thesis, University of California, Berkeley, California, USA,
  1949.

\bibitem{Wang:randomness}
Y.~Wang.
\newblock Resource bounded randomness and computational complexity.
\newblock {\em Theoretical Computer Science}, 237:33--55, 2000.

\end{thebibliography}

\end{document}